%% file: envelope.tex
\documentclass[letterpaper,10pt]{article}
\usepackage[textwidth=11.5cm, textheight=22cm, top=2.0cm, centering]{geometry}
\usepackage{setspace}
\usepackage[T1]{fontenc}
\usepackage{lmodern}
\usepackage[english]{babel}
\usepackage{microtype} % improves typewriter spacing
\usepackage[scaled=0.95]{inconsolata} % true monospace typewriter font
\makeatletter
\renewcommand{\texttt}[1]{{\ttfamily #1}}

\renewcommand{\mathtt}[1]{\text{\texttt{#1}}} % uniform monospace incl. underscores
\makeatother
\usepackage{csquotes}
\usepackage{etoolbox}
\usepackage{bbding}
\usepackage{changepage}
\usepackage{booktabs}
\usepackage{graphicx}
\usepackage[x11names, dvipsnames]{xcolor}
\definecolor{Linkz}{RGB}{30, 110, 170}
\definecolor{Darkenta}{RGB}{185, 35, 90}
\definecolor{Magentz}{RGB}{255, 35, 170}
\definecolor{Lightenta}{RGB}{254, 232, 255}
\definecolor{Reference}{RGB}{35, 180, 90}
\definecolor{Periwinkle}{RGB}{102, 51, 255}
\definecolor{yello}{RGB}{255, 245, 230}
\definecolor{Greeno}{RGB}{0, 140, 100}
\definecolor{Leeno}{RGB}{239, 255, 232}
\definecolor{Nicegreen}{RGB}{100, 200, 130}
\usepackage{amsmath, amsthm, amscd, amssymb}
\usepackage[mathcal]{eucal} % Euler script for \mathcal (P, NP)
\usepackage{stackengine}
\usepackage{bussproofs}
\usepackage[square]{natbib}

\usepackage[
colorlinks=true,
linkcolor=Darkenta,
citecolor=Greeno,
urlcolor=Darkenta,
]{hyperref}
\usepackage{titlesec} % custom section/subsection formatting
\usepackage{abstract} % Abstract
\usepackage{enumitem} % Customizable lists
\usepackage{ccicons} % Creative Commons icons
\usepackage{float}
\usepackage{placeins} %\FloatBarrier
\usepackage{tikz} % TikZ for diagrams
\usepackage{tikz-cd}
\usetikzlibrary{positioning,fit,patterns}
\usetikzlibrary{arrows.meta} % Arrow tips
\usepackage{adjustbox} % For scaling TikZ s
\usepackage{pgfmath} % PGF math engine
\usepackage{ifthen} % Conditional logic in LaTeX
\newtheoremstyle{upright}
{6pt plus 2pt minus 2pt} % Space above
{6pt plus 2pt minus 2pt} % Space below
{\normalfont} % Body font (upright)
{} % Indent amount
{\bfseries} % Theorem head font
{.} % Punctuation after theorem head
{.5em} % Space after theorem head
{} % Theorem head spec
\theoremstyle{upright}
\theoremstyle{upright}
\newtheorem{theorem}{Theorem}[subsection]
\newtheorem{remark}[theorem]{Remark}

\newtheorem{definition}[theorem]{Definition}

\newtheorem{proposition}[theorem]{Proposition}
\newtheorem{lemma}[theorem]{Lemma}

\newtheorem{corollary}[theorem]{Corollary}
\newtheorem{example}[theorem]{Example}
\newtheorem{mechanization}[theorem]{Mechanization}

\makeatletter
\renewenvironment{proof}[1][Proof]{%
	\par\pushQED{\qed}%
	\normalfont
	\topsep6\p@\@plus6\p@\relax
	\trivlist
	\item[\hskip\labelsep\slshape #1\@addpunct{.}]%
}{%
	\popQED\endtrivlist\@endpefalse
}

\makeatother

\usepackage{algpseudocode}
\usepackage[most]{tcolorbox}
\usepackage{listings}
\newtcolorbox{breakbox}[2][]{%
	breakable,
	={#2},
	fonttitle=\bfseries,
	colback=white,
	colframe=black!20,
	coltitle=black,
	colbacktitle=white,
	boxrule=0.5pt,
	arc=0pt,
	boxsep=7pt,
	left=3pt,
	right=2pt,
	top=2pt,
	bottom=4pt,
	fontupper=\small\sffamily, % Applies small mono font to all box content
	#1
}
\AtBeginEnvironment{quotation}{\sffamily}
\renewenvironment{quotation}
{\small\vspace{0.5em}\begin{adjustwidth}{4em}{4em}%

		\setlength{\parindent}{0pt}%
		\setlength{\parskip}{\medskipamount}%
	}
	{\end{adjustwidth}\vspace{1em}}
\newcommand{\customsectionstyle}[2]{%
	\titleformat{\section}[block]
	{\normalfont\fontsize{#1}{1.2\dimexpr#1\relax}\selectfont\centering}
	{\thesection}{1em}%
	{%
		\ifthenelse{\equal{#2}{true}}{\MakeUppercase}{\relax}%
	}%
}
\newcommand{\customsectionspacing}[3]{%
	\titlespacing*{\section}{#1}{#2}{#3}%
}
\newcommand{\customsubsectionstyle}[2]{%
	\titleformat{\subsection}[block]
	{\normalfont\fontsize{#1}{1.2\dimexpr#1\relax}\selectfont\centering}
	{\thesubsection}{1em}%
	{%
		\ifthenelse{\equal{#2}{true}}{\MakeUppercase}{\relax}%
	}%
}
\newcommand{\customsubsectionspacing}[3]{%
	\titlespacing*{\subsection}{#1}{#2}{#3}%
}
\customsectionstyle{14pt}{true}
\customsectionspacing{0pt}{3.5ex plus 1ex minus 0.2ex}{2.5ex}
\customsubsectionstyle{11pt}{true}
\customsubsectionspacing{0pt}{4ex plus 1ex minus 0.2ex}{2ex}
\usepackage{caption}

\makeatletter
\newcommand{\shorttitle}[1]{\def\@shorttitle{#1}}
\newcommand{\email}[1]{\def\@email{#1}}
\newcommand{\metadata}[1]{\def\@metadata{#1}}
\renewcommand{\maketitle}{%
	\begin{center}
		\vfill
		{\fontsize{18pt}{19pt}\selectfont \@title \par}
		\vspace{1em}
		{\normalsize \@author \par}
		\vspace{0.1em}
		{\normalsize \@date \par}
	\end{center}
}
\makeatother

\begin{document}
\input{content}
\end{document}

%% file: content.tex
% === Notation ===
\newcommand{\sem}[2]{[\![#1]\!](#2)}
\newcommand{\code}[1]{\ulcorner#1\urcorner}
\newcommand{\pair}[2]{\langle#1,#2\rangle}
\newcommand{\lvl}{\mathsf{lvl}}
\newcommand{\Acc}{\mathsf{Acc}}
\newcommand{\BF}{\mathsf{BF}}
\newcommand{\SAT}{\mathsf{SAT}}
\providecommand{\Defeq}{\mathrel{\triangleq}}
\newcommand{\cneg}{c_{\neg}}
\newcommand{\cdup}{c_{d}}
\newcommand{\cex}{c_{\exists}}
\newcommand{\fneg}{\mathsf{neg}}
\newcommand{\fdup}{\mathsf{dup}}
\newcommand{\fproj}{\mathsf{proj}}
\newcommand{\wit}{\mathsf{wit}}
\newcommand{\spec}{\mathsf{sp}}
\newtheorem*{maintheorem}{Theorem}

% === TITLE ===
\title{\uppercase{The Solver's Paradox for SAT}}
\author{Milan Rosko}
\date{November 2026}

\maketitle

\begin{center}\footnotesize{
	ORCID: \href{https://orcid.org/0009-0003-1363-7158}{\footnotesize\textsf{0009-0003-1363-7158}}\\
}
\end{center}

\begin{abstract}\vspace{-1ex}\footnotesize{
We show in \textsc{Rocq} that every circuit verifier $v$ determines a formula $\psi_v$ whose satisfiability is $v$'s acceptance of its own diagonal claim: a satisfiability predicate thus behaves like a provability predicate. Every solver represented in the language misclassifies its own diagonal instance, and every verifier misjudges its own diagonal claim. We distinguish our result from (i) the refuters of Chen et al.\ (2024), which produce from the description of a solver an instance it gets wrong, and (ii) the refuter of Gutfreund et al.\ (2007) for efficient SAT solvers, which needs a hardness hypothesis and a simulation. Here the refuter is purely syntactic, and it separates evaluation of the language from its own expressions, not SAT from a resource-bounded class. The identity behind the fixed point uses no property of quotation; its reading as self-reference depends on the coding, and an injective recoding flips the branch. A truth table decides $\psi_v$ only from outside the language. Written inside, it is brute force: it settles the claim from above and has a diagonal instance of its own. This yields an unconditional certification hierarchy with tower growth, a quadratic lower bound on every internal evaluator, and an exponential lower bound on brute force.
}\end{abstract}

\section{Introduction}
\label{sec:introduction}

\begin{figure}[H]
	\centering
	\includegraphics[width=\textwidth]{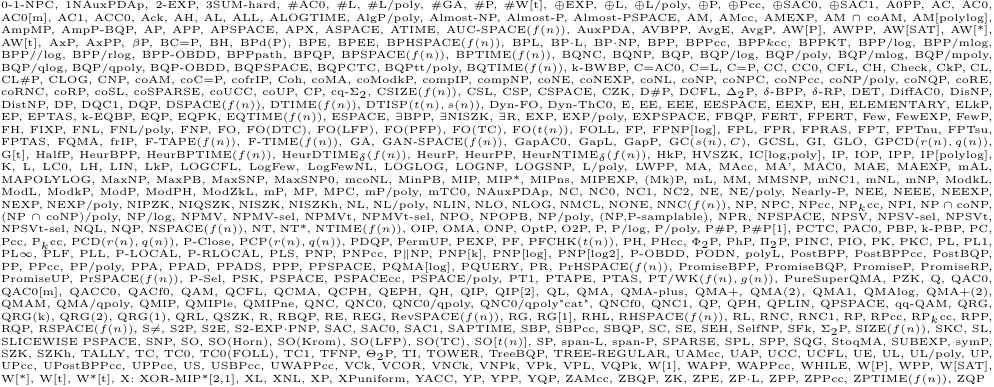}
	\caption{A word cloud, by the author, of roughly the classes listed in the \textsc{Complexity Zoo} of \citet{complexityzoo}, as of 2025. Every class is fixed by a machine, a reduction or a proof, and each of these is an instance of the problems the map classifies. Section~\ref{sec:discussion:local} asks what follows if the map is written in the language it surveys.}
\label{fig:complexity-zoo}
\end{figure}

\subsection{Solvability as Provability}
\label{sec:introduction:bew}

Satisfiability is certified by verifiers. A verifier $v$ \emph{accepts} a claim when some witness makes it hold,
\begin{equation}
	\Acc_v(e,x)\;\Defeq\;\exists w\,\sem v{\pair{\pair{\code e}x}{w}}.
\end{equation}
Consider this in the form of Gödel's predicate $\mathrm{Bew}(x)\Defeq\exists p\,\mathrm{Prf}(p,x)$ \citep{goedel31}, a decidable relation under one existential, as verifiers are proof systems in the sense of \citet{cookreckhow79}, and their witnesses are proofs. Applied to non-acceptance by $v$, the \textsc{Diagonal Lemma} \citep{boolos} is realized by the grammar as a single term,
\begin{equation}
	g_v\Defeq\neg\,\fdup(\fproj(v)),
\end{equation}
whose claim $G_v$ at its own quotation satisfies $G_v\leftrightarrow\neg\Acc_v$. For a circuit verifier, acceptance at that claim is satisfiability of a propositional formula $\psi_v$ no larger than $v$, and so
\begin{equation}
	G_v\;\leftrightarrow\;\neg\SAT(\psi_v).
\end{equation}
the above is Gödel's sentence with satisfiability in place of provability. Turned on itself, solvability behaves like $\mathrm{Bew}$. Gödel, in his 1956 letter to von Neumann \citep[translated in][]{sipser92}, read bounded provability as a search problem; here search is read as bounded provability, in the setting of \citet{parikh71} and \citet{pudlak86}, with propositional logic as the language and the size of the verifier as the bound.

\begin{table}[H]
\centering
\renewcommand{\arraystretch}{1.25}
\small
\begin{adjustbox}{max width=\textwidth}
\begin{tabular}{@{}lll@{}}
	\toprule
	& Incompleteness & Solver's Paradox \\
	\midrule
	Predicate & $\mathrm{Bew}(x)\Defeq\exists p\,\mathrm{Prf}(p,x)$ & $\Acc_v\Defeq\exists w\,\sem v{\pair{\pair{\code{g_v}}{\code{g_v}}}w}$ \\
	Witness & proof $p$ & witness stream $w$ \\
	Checker & decidable relation $\mathrm{Prf}$ & circuit verifier $v$ \\
	Fixed point & $G\leftrightarrow\neg\mathrm{Bew}(\code G)$ & $G_v\leftrightarrow\neg\SAT(\psi_v)$ \\
	Witness space & unbounded & finite, at most $2^{|v|}$ assignments \\
	Status of the sentence & unprovable in $\mathrm{T}$ if $\mathrm{T}$ is consistent & decided by a truth table, from outside \\
	Under soundness & true and unprovable & true and unaccepted \\
	Settled by & $\mathrm{T}+\mathrm{Con}(\mathrm{T})$ & any evaluator for level $\lvl(g_v)$, e.g.\ $\BF_{\lvl(g_v)}$ \\
	Derivability conditions & D1--D3 proved for $\mathrm{Bew}_{\mathsf{PA}}$ & stipulated, jointly satisfiable (Section~\ref{sec:internal}) \\
	\bottomrule
\end{tabular}
\end{adjustbox}
\caption{Provability and satisfiability at the diagonal. The two agree at the fixed point; the bound on the witness space is what makes the right-hand sentence decidable.}
\label{tab:bew-sat}
\end{table}

The correspondence is fragile at one place only: Gödel's sentence is undecidable in its theory because proofs are unbounded; $\psi_v$ has at most $2^{|v|}$ assignments, and its truth table decides $G_v$.

We can however translate the solver into the language as well, such that the diagonal applies to it. Every term $S$ is total, returning a verdict at every coded instance, and the term $d_S\Defeq\neg\,\fdup(S)$ satisfies
\begin{equation}
	\sem{d_S}{\code{d_S}}\;\leftrightarrow\;\neg\,\sem S{\pair{\code{d_S}}{\code{d_S}}}.
\end{equation}
Now every represented solver misclassifies its own diagonal instance. The argument is trivial, a single diagonal step, but it places brute force among the expressions it evaluates. This is the \emph{classification} half of the paradox, the \emph{certification} half adds projection: for every represented verifier, acceptance and truth cannot agree at $G_v$, and for circuit verifiers the misjudged claim is $\neg\SAT(\psi_v)$, an ordinary propositional instance. That they disagree is proved outright; in which direction is a decision about acceptance, and for circuit verifiers it is the truth table that makes it. In neither half is the semantics of satisfiability touched. What fails is representation: a solver or verifier that is one of the expressions has an instance it gets wrong. If the truth table is one of the formulas, it is wrong about one of them; the table that is right about all of them is the position one level up.

\subsection{Settlement and the Zoo}
\label{sec:introduction:zoo}

Take “brute force” as a term $\BF_n$ that evaluates every expression up to level $n$, and so settles $G_v$. Yet $\BF_n$ is one more expression of the language, with a diagonal claim of its own. Settlement therefore ascends: each stage settles the claim below it and fails at a fresh claim of its own, as in the progressions of \citet{turing39} and \citet{feferman62}. The ascent has a price measured in bits. Every evaluator for a level at or above $\lvl(g_v)$ lies above $v$, and every evaluator for level $n$ has level at least $(n-4)(n+2)$.

The map $v\mapsto g_v$ has the shape of a productive function in the sense of \citet{post44}. If $v$ is sound at its claim, then $G_v$ is true and not accepted by $v$. So from any sound verifier, given as an expression, nine bits of syntax produce a truth that it misses. The true diagonal claims behave like the complement of a creative set. Every one of them is settled by some certifier of the language, no certifier settles them all, and every internal survey of the space comes with an effective way past it. If the map of Figure~\ref{fig:complexity-zoo} is read as such a survey, it cannot be completed from inside (Section~\ref{sec:discussion:local}).

The limits are part of the result. The correspondence with $\mathrm{Bew}$ is exact at the fixed point; forming the fixed point needs no derivability condition, while the theorems of Löb and the second incompleteness theorem need all three. Section~\ref{sec:internal} stipulates them as premises, as abstract provability logic does, and shows that the premises are jointly satisfiable; it does not derive them for the projection of a concrete verifier. Every result holds relative to every oracle. This does not set $\mathcal{P}$ versus $\mathcal{NP}$ aside: it places the constructions on one side of an analogue of the barrier of \citet{bakergillsolovay75}, and Section~\ref{sec:discussion:about} observes that the barrier has the shape of the ascent itself.

\subsection{The Result}
\label{sec:introduction:result}

Work in a small language of Boolean terms over bit streams, with duplication and witness projection (Section~\ref{sec:sat}). A \emph{circuit verifier} is a term without projection.

The proof is short and elementary. Its interest lies in three facts that a diagonal argument usually states as one. The \emph{identity}
\begin{equation}\label{eq:identity}
	\sem{g_v}{x}\;\leftrightarrow\;\neg\exists w\,\sem v{\pair{\pair xx}{w}}
\end{equation}
holds at every datum and uses no property of quotation; it belongs to the language-first tradition of \citet{smullyan57,smullyan94} and to the fixed-point pattern of \citet{lawvere69}. The \emph{reading} of \eqref{eq:identity} as self-reference enters only with the choice $x=\code{g_v}$, and it depends on the coding, as in \citet{halbachvisser14a,halbachvisser14b}: an injective recoding flips the branch. The \emph{settlement} happens one level up, as described above. Typed self-interpretation shows how much rests on diagonal duplication: typing blocks it in \citet{brownpalsberg16}, and it returns in \citet{dingzhang24}.

The construction belongs to a line of results that produce, from the description of a solver, an instance it gets wrong. \citet{gutfreundshaltieltashma07} do this for SAT under the hypothesis $\mathcal{NP}\not\subseteq\mathcal{BPP}$: by \textsc{Cook--Levin} they write the failure of the solver as a formula of polynomial size, and obtain at most three formulas, on one of which it errs. \citet{chenjinsanthanamwilliams24} call a separation \emph{constructive} when such a \emph{refuter} exists. Here the refuter is syntax. It needs no hypothesis and no simulation, and the formula $\psi_v$ is $v$ itself with its diagonal claim substituted, no larger than $v$. The price is the problem: what is separated is evaluation of the language from its own expressions, as in a hierarchy theorem, not SAT from a resource-bounded class, and SAT enters only as the form of the misjudged claim. Where the refuter of \citet{gutfreundshaltieltashma07} narrows the error to three formulas without naming one, here the instance is named and the direction of the error is not: that acceptance and truth disagree is constructive, and which way they disagree is the decision of $\SAT(\psi_v)$.

(i) For every circuit verifier we construct the formula $\psi_v$, prove that its satisfiability is the verifier's acceptance at its diagonal claim and its unsatisfiability the truth of that claim, and give a CNF form with at most $3|v|+1$ clauses; all of this holds relative to every oracle. (ii) We isolate the identity \eqref{eq:identity}, which uses only the formation laws, prove its classification counterpart, that every represented solver misclassifies its own diagonal instance, so that $u\mapsto d_u$ is a refuter of six bits, and show that the two-branch dichotomy is equivalent to decided acceptance and, demanded of every verifier complete at its claim, to excluded middle. (iii) We show that the obstruction is invariant under faithful interpretation and under every quotation, while its branch flips between two injective quotations. (iv) We construct brute force as a term $\BF_n$ for every level, obtain an unconditional certification hierarchy with tower growth, and prove that every internal evaluator for level $n$ has level at least $(n-4)(n+2)$. (v) Under the derivability conditions, stipulated as premises, the same sentence is an internal Gödel fixed point, and the internal premises are jointly satisfiable.

\subsection{Roadmap}
\label{sec:introduction:roadmap}

Section~\ref{sec:sat} defines the language and proves the \textsc{Solver's Paradox}. Section~\ref{sec:three} abstracts the proof into identity, reading and settlement. Section~\ref{sec:brute} builds brute force inside the language and derives the hierarchy and the bounds. Section~\ref{sec:coding} treats coding and oracles, and Section~\ref{sec:internal} the internal reading. Section~\ref{sec:mech} maps every result to its \textsc{Rocq} endpoint, and Section~\ref{sec:discussion} says what the results do and do not say.

\section{The Solver's Paradox}
\label{sec:sat}

\subsection{The Language}
\label{sec:sat:language}

\begin{example}
\label{ex:heuristic}
	We start with an example. Code a formula over one variable $p$ by two bits, saying that the unit clause $p$, respectively $\neg p$, is present. The unit rule, the case of maximum occurrences in minimum-sized clauses, sets $p$ true if $p$ is a clause, false if $\neg p$ is, and otherwise takes a witness bit; it accepts when every clause present is satisfied. At the diagonal datum the claim bits sit at positions $0$ and $4$ and the witness bit at position $1$, by \eqref{eq:diagonal-datum}, so with $p\Defeq\mathsf{bit}_0\lor(\neg\mathsf{bit}_4\land\mathsf{bit}_1)$ the rule is the circuit verifier
	\begin{equation}
		v\;\Defeq\;(\neg\mathsf{bit}_0\lor p)\land(\neg\mathsf{bit}_4\lor\neg p)
	\end{equation}
	of $20$ nodes. Fed its own diagonal claim, $v$ reads the first two bits of $\code{g_v}$, which open the tag of the leading negation: $0$, so $p$ is not a clause, and some bit $\beta$. Specialization gives
	\begin{equation}
		\psi_v\;=\;(\neg0\lor p')\land(\neg\beta\lor\neg p'),
		\quad
		p'\Defeq0\lor(\neg\beta\land p_0),
	\end{equation}
	which $p_0=0$ satisfies. The rule sets $p$ false and accepts, so by Theorem~\ref{thm:solvers-paradox} it certifies the false claim that it accepts no witness for $g_v$.
\end{example}

\begin{definition}[Terms, data, quotation]
\label{def:concrete}
Terms are
\begin{equation}
	t\;::=\;\mathsf{c}_b\mid\mathsf{bit}_i\mid\neg t\mid t\land t\mid t\lor t\mid\fdup(t)\mid\fproj(t)\mid\mathsf{O}_\ell,
\end{equation}
with $b$ a Boolean, $i\in\mathbb N$ and $\ell$ a finite list of positions. Data are bit streams $x:\mathbb N\to\{0,1\}$. Pairing interleaves:
\begin{equation}
	\pair xy_{2k}=x_k,\quad \pair xy_{2k+1}=y_k.
\end{equation}
Given an oracle $O$, $\mathsf{c}_b$ holds when $b=1$, $\mathsf{bit}_i$ holds at $x$ when $x_i=1$, the connectives are classical, and
\begin{equation}
	\sem{\fdup(t)}x\leftrightarrow\sem t{\pair xx},
	\quad
	\sem{\fproj(t)}x\leftrightarrow\exists w\,\sem t{\pair xw},
\end{equation}
while $\mathsf O_\ell$ holds when $O$ holds of the bits of $x$ at $\ell$. Each term has a prefix-free serialization, with a three-bit tag per node and indices in unary. Its quotation $\code t$ is that serialization padded with zeros, and its level $\lvl(t)$ is the length of the serialization. We write $|t|$ for the node count.
\end{definition}

\begin{remark}
Semantics is proposition-valued, because projection quantifies over infinite streams. On terms without projection every claim is decided by finitely many bits.
\end{remark}

\begin{definition}[Acceptance, diagonal sentence, circuit verifier]
\label{def:diagonal}
A term $v$ \emph{accepts} the claim that $e$ holds at $x$ when some witness stream makes $v$ hold of the claim:
\begin{equation}
	\Acc_v(e,x)\;\Defeq\;\exists w\,\sem v{\pair{\pair{\code e}x}{w}}.
\end{equation}
The \emph{diagonal sentence} of $v$ is $g_v\Defeq\neg\,\fdup(\fproj(v))$. Its \emph{diagonal claim} is $G_v\Defeq\sem{g_v}{\code{g_v}}$, and we abbreviate $\Acc_v\Defeq\Acc_v(g_v,\code{g_v})$. A \emph{circuit verifier} is a term with neither projection nor oracle nodes.
\end{definition}

\subsection{The Fixed Point}
\label{sec:sat:fixed}

\begin{lemma}[Identity]
\label{lem:concrete-identity}
For every term $v$ and stream $x$, identity \eqref{eq:identity} holds.
\end{lemma}

\begin{proof}
Unfold negation, duplication and projection in turn:
\begin{equation*}
	\begin{aligned}
	\sem{g_v}x&\;\leftrightarrow\;\neg\sem{\fdup(\fproj(v))}x\;\leftrightarrow\;\neg\sem{\fproj(v)}{\pair xx}\\
	&\;\leftrightarrow\;\neg\exists w\,\sem v{\pair{\pair xx}w}.
	\end{aligned}\qedhere
\end{equation*}
\end{proof}

\begin{corollary}[Fixed point]
\label{cor:concrete-fixed}
For every term $v$, $G_v\leftrightarrow\neg\Acc_v$.
\end{corollary}

\begin{proof}
Put $x=\code{g_v}$ in Lemma~\ref{lem:concrete-identity}.
\end{proof}

\subsection{Specialization}
\label{sec:sat:specialization}

At the diagonal datum $z=\pair{\pair{\code{g_v}}{\code{g_v}}}w$ the even positions are fixed by the quotation and the odd positions belong to the witness:
\begin{equation}\label{eq:diagonal-datum}
	z_{2k}=\code{g_v}_{\lfloor k/2\rfloor},
	\quad
	z_{2k+1}=w_k.
\end{equation}
A circuit verifier reads finitely many of these positions. Replacing each claim bit by its constant and each witness bit $w_k$ by a propositional variable $p_k$ turns $v$ into a formula.

\begin{definition}[Specialized formula]
\label{def:specialization}
For a circuit term $t$ and an assignment $a$ of propositional formulas to positions, let
\begin{equation}
	\begin{aligned}
	&\spec(\mathsf c_b,a)=b,\quad \spec(\mathsf{bit}_i,a)=a(i),\quad \spec(\neg t,a)=\neg\spec(t,a),\\
	&\spec(t\circ u,a)=\spec(t,a)\circ\spec(u,a)\quad(\circ\in\{\land,\lor\}),\\
	&\spec(\fdup(t),a)=\spec(t,\,i\mapsto a(\lfloor i/2\rfloor)).
	\end{aligned}
\end{equation}
Let $a_v(2k)$ be the constant $\code{g_v}_{\lfloor k/2\rfloor}$ and $a_v(2k+1)=p_k$, and put $\psi_v\Defeq\spec(v,a_v)$.
\end{definition}

\begin{lemma}[Specialization]
\label{lem:specialization}
For every circuit verifier $v$,
\begin{equation}
	\Acc_v\leftrightarrow\SAT(\psi_v),
	\quad
	|\psi_v|\leq|v|,
	\quad
	3|\psi_v|\leq\lvl(v).
\end{equation}
\end{lemma}

\begin{proof}
By induction on $t$: if an environment $\rho$ evaluates each $a(i)$ to $x_i$, then $\sem tx$ holds exactly when $\rho$ satisfies $\spec(t,a)$. The duplication case uses $\pair xx_i=x_{\lfloor i/2\rfloor}$, and the oracle is never consulted. By \eqref{eq:diagonal-datum}, an environment for $\psi_v$ evaluates $a_v$ to the diagonal datum of the witness $w_k\Defeq\rho(p_k)$, and every witness arises this way, so $\Acc_v$ is satisfiability of $\psi_v$. The size bound holds because $\spec$ copies each non-duplication node at most once and drops duplications; $3|v|\leq\lvl(v)$ because every node carries a three-bit tag.
\end{proof}

\subsection{The Theorem}
\label{sec:sat:theorem}

\begin{theorem}[\textsc{Solver's Paradox}]
\label{thm:solvers-paradox}
For every oracle and every circuit verifier $v$:
\begin{enumerate}[label=(\roman*)]
	\item $|\psi_v|\leq|v|$, and the truth table of $\psi_v$ decides $\SAT(\psi_v)$;
	\item $\Acc_v\leftrightarrow\SAT(\psi_v)$ and $G_v\leftrightarrow\neg\SAT(\psi_v)$;
	\item exactly one of the following holds:
	\begin{equation}
		\SAT(\psi_v)\land\Acc_v\land\neg G_v,
		\quad
		\neg\SAT(\psi_v)\land\neg\Acc_v\land G_v.
	\end{equation}
\end{enumerate}
\end{theorem}

\begin{proof}
Item (i) is Lemma~\ref{lem:specialization} together with the correctness of truth tables. Item (ii) combines Lemma~\ref{lem:specialization} with Corollary~\ref{cor:concrete-fixed}. For (iii), the truth table decides $\SAT(\psi_v)$. If $\psi_v$ is satisfiable, then $v$ accepts and, by the fixed point, $G_v$ fails. If it is unsatisfiable, then $v$ does not accept and $G_v$ holds. The two cases disagree on $\SAT(\psi_v)$.
\end{proof}

\begin{remark}
	For a term $S$ and a term $v$, let $d_S\Defeq\fneg(\fdup(S))$ and $g_v\Defeq\fneg(\fdup(\fproj(v)))$ be the problems...
	\begin{quotation}
		Does $S$ reject the claim of $d_S$?\\
		Does $v$ accept no witness for the claim of $g_v$?
	\end{quotation}
	...then $S$ misclassifies $d_S$ (Theorem~\ref{thm:classification-escape}) and $v$ misjudges $g_v$ (Corollary~\ref{cor:pointwise}). The first is the pattern of Tarski's undefinability \citep{tarski36}, the second that of Gödel's sentence; projection is the step from one to the other. If $v$ is a circuit verifier, the claim of $g_v$ is the unsatisfiability of $\psi_v$, a formula no larger than $v$.

	The truth table, standing “outside” the language, returns the correct verdict on $\psi_v$. The verifier $v$, whose acceptance $\psi_v$ represents, either accepts a witness for a false claim or accepts none for a true one. The asymmetry is one of position, not of method. Written inside the language, the truth table is the term $\BF_n$ of Section~\ref{sec:brute}; for $n\geq\lvl(g_v)$ it classifies $g_v$ correctly, and it misclassifies its own diagonal instance $d_{\BF_n}$.
\end{remark}

\begin{example}[Minimal branches]
\label{ex:minimal}
The constant $\mathsf c_1$ accepts a false claim, and $\psi_{\mathsf c_1}$ is the satisfiable constant. The constant $\mathsf c_0$ leaves a true claim uncertified, and $\psi_{\mathsf c_0}$ is unsatisfiable. Neither branch is decorative.
\end{example}

\begin{corollary}[CNF form]
\label{cor:cnf}
By the transformation of \citet{tseitin68}, $\Acc_v$ is the satisfiability of a CNF with at most $3|v|+1$ clauses, each of width at most three, and at most $|\psi_v|$ fresh variables.
\end{corollary}

The diagonal claim is not undecidable. A finite table with at most $2^{|v|}$ rows decides it, and decides it against $v$. The table does so from outside the language; inside, it would be one more term, with a diagonal instance of its own.

\section{Identity, Reading, Settlement}
\label{sec:three}

The proof of Theorem~\ref{thm:solvers-paradox} used three formation laws, a quotation map, and a decision. This section isolates each, in a setting where semantics is arbitrary.

\subsection{Diagonal Structures}
\label{sec:three:structures}

\begin{definition}[Diagonal and certification structures]
\label{def:structures}
A \emph{diagonal structure} consists of expressions and data, a proposition-valued semantics $\sem ex$, a quotation map $e\mapsto\code e$, a pairing $\pair xy$, a level $\lvl(e)\in\mathbb N$, and operations $\fneg$, $\fdup$ with
\begin{equation}
	\sem{\fneg(e)}x\leftrightarrow\neg\sem ex,
	\quad
	\sem{\fdup(e)}y\leftrightarrow\sem e{\pair yy},
\end{equation}
$\lvl(\fneg(e))\leq\lvl(e)+\cneg$ and $\lvl(\fdup(e))\leq\lvl(e)+\cdup$. A \emph{certification structure} adds a map $\wit$ from witnesses to data and an operation $\fproj$ with $\sem{\fproj(e)}z\leftrightarrow\exists w\,\sem e{\pair z{\wit(w)}}$ and $\lvl(\fproj(e))\leq\lvl(e)+\cex$. Acceptance, $g_v=\fneg(\fdup(\fproj(v)))$, $G_v$ and $\Acc_v$ are defined as in Definition~\ref{def:diagonal}. Costs are \emph{exact} when the level inequalities are equalities.
\end{definition}

\begin{remark}
Quotation is a bare map: no decoding, no injectivity, and no relation between $\code e$ and the meaning of $e$ is assumed. ``Accepts'' is not ``holds''; soundness and completeness are separate hypotheses, used only where stated.
\end{remark}

An expression $u$ \emph{evaluates level $n$} if $\sem u{\pair{\code e}x}\leftrightarrow\sem ex$ whenever $\lvl(e)\leq n$. A verifier is \emph{sound} (\emph{complete}) through level $n$ if $\Acc_v(e,x)$ implies (is implied by) $\sem ex$ whenever $\lvl(e)\leq n$.

\subsection{The Identity and the Reading}
\label{sec:three:identity}

\begin{theorem}[Coding-free identity and fixed point]
\label{thm:fixed-point}
In every certification structure, identity \eqref{eq:identity} holds at every datum, with $\wit(w)$ in place of $w$, and $G_v\leftrightarrow\neg\Acc_v$.
\end{theorem}

\begin{proof}
As in Lemma~\ref{lem:concrete-identity} and Corollary~\ref{cor:concrete-fixed}; only the formation laws are used.
\end{proof}

The identity is a statement about one verifier at one datum, and nothing in it is about itself. The reading enters with the choice $x=\code{g_v}$.

\begin{corollary}[Pointwise incompleteness]
\label{cor:pointwise}
For every verifier, $\neg(\Acc_v\leftrightarrow G_v)$. If $v$ is sound at its diagonal claim, then $G_v\land\neg\Acc_v$.
\end{corollary}

\begin{proof}
From $\Acc_v\leftrightarrow G_v$ and the fixed point, $\Acc_v\leftrightarrow\neg\Acc_v$. If $\Acc_v\to G_v$, then $\Acc_v$ gives $G_v$ and so $\neg\Acc_v$; hence $\neg\Acc_v$, and $G_v$ follows.
\end{proof}

No decision was needed. The verifier is wrong at its claim in the weak sense that acceptance and truth are not equivalent; saying \emph{which way} it is wrong is a further step.

\begin{proposition}[Faithful invariance]
\label{prop:faithful}
An interpretation that transports expressions and data, commutes with $\fneg$, $\fdup$, $\fproj$, quotation and pairing, and preserves semantics, sends $g_v$ to the diagonal sentence of the image of $v$ and preserves the diagonal claim and its non-acceptance. Where acceptance is decided on both sides, it also preserves acceptance and the branch.
\end{proposition}

Non-acceptance transports outright; acceptance transports only where it is decided, because $\neg A\leftrightarrow\neg B$ does not yield $A\leftrightarrow B$. Theorem~\ref{thm:coding} shows that the restriction to faithful interpretations is necessary.

\subsection{The Decision}
\label{sec:three:decision}

\begin{theorem}[Certification Dichotomy]
\label{thm:dichotomy}
For every verifier,
\begin{equation}
	(\Acc_v\lor\neg\Acc_v)
	\;\leftrightarrow\;
	(\Acc_v\land\neg G_v)\lor(\neg\Acc_v\land G_v),
\end{equation}
and the two branches exclude each other. Boolean semantics decides acceptance for every verifier.
\end{theorem}

\begin{proof}
Each branch decides $\Acc_v$. Conversely, $\Acc_v$ gives $\neg G_v$ and $\neg\Acc_v$ gives $G_v$. The branches disagree on $\Acc_v$.
\end{proof}

In Theorem~\ref{thm:solvers-paradox} the truth table supplied the decision. In general it is one decision at one proposition, not decidability of the language, and it cannot be weakened. Corollary~\ref{cor:pointwise} needed no decision: that acceptance and truth disagree is constructive, and only the direction of the disagreement is information. For a circuit verifier that information comes from the truth table of $\psi_v$, that is, from brute force, the extension of $v$ one level up, and not from $v$.

\begin{theorem}[Sharpness]
\label{thm:sharpness}
The dichotomy, demanded in every certification structure of every verifier complete at its diagonal claim, is equivalent to excluded middle.
\end{theorem}

\begin{proof}
	Excluded middle decides $\Acc_v$. Conversely, given a proposition $Q$, take a structure with one datum, one witness, and expressions generated freely from a single base expression meaning $Q\lor\neg Q$; $\fneg$ is negation, $\fdup$ and $\fproj$ are transparent, and all levels and costs are zero. The base verifier accepts its claim exactly when $Q\lor\neg Q$. Since $\neg\neg(Q\lor\neg Q)$, its claim is refutable, so it is complete at that claim vacuously. The dichotomy for it decides $Q\lor\neg Q$.
\end{proof}

\begin{remark}
	\label{rem:boundary}
	Completeness at the claim is the unhappy demand. From $G\leftrightarrow\neg A$ and $G\to A$ one gets $\neg G$ and $\neg\neg A$: a verifier complete at its claim reaches the branch of a false acceptance only up to double negation, and the dichotomy collapses to $A$ itself. Soundness at the claim, by contrast, reaches the other branch outright: an uncertified truth, and nothing inconsistent.
\end{remark}

\subsection{Settlement}
\label{sec:three:settlement}

An expression $u$ \emph{classifies} $e$ at $x$ correctly when $\sem u{\pair{\code e}x}\leftrightarrow\sem ex$. Every expression has a semantics at every datum, so every expression, read as a solver, is total.

\begin{theorem}[Classification escape]
\label{thm:classification-escape}
For every expression $u$, put $d_u\Defeq\fneg(\fdup(u))$. Then
\begin{equation}
	\sem{d_u}{\code{d_u}}\;\leftrightarrow\;\neg\,\sem u{\pair{\code{d_u}}{\code{d_u}}},
\end{equation}
so $u$ misclassifies $d_u$ at its own quotation, and no expression classifies every expression correctly at every datum.
\end{theorem}

\begin{proof}
The equivalence is the formation laws for $\fneg$ and $\fdup$. Correct classification at $\code{d_u}$ would give $\sem u{\pair{\code{d_u}}{\code{d_u}}}\leftrightarrow\neg\sem u{\pair{\code{d_u}}{\code{d_u}}}$, which no proposition satisfies.
\end{proof}

This is the internal form of the statement that every total solver misclassifies an instance. It needs no bound, no decidability and no property of quotation. Its resource-bounded, conditional form is the theorem of \citet{gutfreundshaltieltashma07} (Section~\ref{sec:introduction:result}). The external truth table is not a counterexample, because it is not an expression; every expression that realizes it, $\BF_n$ included, is subject to the theorem. Evaluator Ascent records the level at which the failure occurs.

\begin{lemma}[Evaluator ascent]
\label{lem:evaluator-ascent}
If $u$ evaluates level $n$, then $n<\lvl(u)+\cdup+\cneg$. If $v$ is sound and complete through level $n$, then $n<\lvl(v)+\cex+\cdup+\cneg$.
\end{lemma}

\begin{proof}
Let $d=d_u$ and suppose $\lvl(u)+\cdup+\cneg\leq n$, so that $\lvl(d)\leq n$ and $u$ would classify $d$ correctly, contrary to Theorem~\ref{thm:classification-escape}. Explicitly, $\sem u{\pair{\code d}{\code d}}\leftrightarrow\sem d{\code d}\leftrightarrow\neg\sem u{\pair{\code d}{\code d}}$, which no proposition satisfies. For the second claim, soundness and completeness make $\fproj(v)$ evaluate level $n$.
\end{proof}

This is the graded form of the pattern of \citet{lawvere69}.

\begin{theorem}[Settled one level up]
\label{thm:settled}
Every evaluator $B$ for a level $n\geq\lvl(g_v)$ settles the diagonal claim:
\begin{equation}
	\sem B{\pair{\code{g_v}}{\code{g_v}}}\;\leftrightarrow\;G_v.
\end{equation}
Under exact costs, $\lvl(g_v)=\lvl(v)+\cex+\cdup+\cneg$, and every such $B$ satisfies $\lvl(v)+\cex<\lvl(B)$.
\end{theorem}

\begin{proof}
The first claim is the evaluator property at $g_v$. For the second, Lemma~\ref{lem:evaluator-ascent} gives $\lvl(g_v)\leq n<\lvl(B)+\cdup+\cneg$; substitute the exact level of $g_v$.
\end{proof}

A family $(V_n)_n$ is a \emph{bounded certification family} if each $V_n$ is sound and complete through level $n$. Fix $n_0$ and put $n_{k+1}\Defeq\lvl(g_{V_{n_k}})$.

\begin{theorem}[Ascent hierarchy]
\label{thm:hierarchy}
For a bounded certification family and every $k$: $n_k<n_{k+1}$; $V_{n_k}$ is not correct at its own diagonal claim; and $V_{n_{k+1}}$ is correct there. Moreover $n_{k+1}\leq\lvl(V_{n_k})+\cex+\cdup+\cneg$, $n_0+k\leq n_k$, and the diagonal sentences $g_{V_{n_k}}$ are pairwise distinct.
\end{theorem}

\begin{proof}
If $n_{k+1}\leq n_k$, then $V_{n_k}$ would be correct at its own claim, contrary to Corollary~\ref{cor:pointwise}. Correctness of $V_{n_{k+1}}$ there is its soundness and completeness at level $n_{k+1}$. The bounds follow from the level inequalities and strict growth; distinct levels give distinct sentences.
\end{proof}

Each stage settles the claim below it and fails at a fresh claim of its own. The hierarchy tracks a moving boundary, not one unresolved sentence.

\section{Brute Force Inside the Language}
\label{sec:brute}

\subsection{Grading}
\label{sec:brute:grading}

\begin{theorem}[Concrete grading]
\label{thm:grading}
The language of Section~\ref{sec:sat} is a certification structure for every oracle. Serialization parses uniquely; negation, duplication and projection cost exactly three bits, so $\lvl(g_v)=\lvl(v)+9$; and if the expressions of level at most $n$ are closed under diagonal formation, there are none.
\end{theorem}

\subsection{Brute Force}
\label{sec:brute:bf}

For a term $e$ let $\mathsf{match}(e)$ be the conjunction of literals asserting that the even input bits spell the serialization of $e$, and let $e^{\mathrm{odd}}$ be $e$ reindexed to read the odd bits. Put
\begin{equation}
	\BF_n\;\Defeq\;\bigvee_{\lvl(e)\leq n}\bigl(\mathsf{match}(e)\land e^{\mathrm{odd}}\bigr).
\end{equation}
Unique parsing makes exactly the intended disjunct fire on $\pair{\code e}x$. Reindexing absorbs duplication into the index map and keeps projection as a node with a shifted map, so $\BF_n$ handles the full language, not only its circuits. Reading only the claim half of a certification datum turns $\BF_n$ into a certifier $V_n$.

\begin{theorem}[Brute-force evaluation]
\label{thm:bf}
For every oracle, $\BF_n$ evaluates level $n$, and $3^{n+1}\leq\lvl(\BF_n)+3$. Every evaluator $B$ for level $n$ satisfies $n<\lvl(B)+6$. The certifiers $(V_n)_n$ form a bounded certification family.
\end{theorem}

\begin{corollary}[The SAT claim, settled]
\label{cor:sat-settled}
For every circuit verifier $v$, the evaluator $\BF_{\lvl(v)+9}$ settles $G_v$, and every evaluator for a level $n\geq\lvl(g_v)$ has level above $\lvl(v)+3$.
\end{corollary}

\begin{proof}
Theorems~\ref{thm:grading}, \ref{thm:bf} and~\ref{thm:settled}, with $\cex=3$.
\end{proof}

The truth table settled the claim from outside the language. Brute force is that table written inside it: the verifier is extended by a term of the same language, which settles $G_v$ from a level far above $g_v$. The extension has a diagonal of its own, $\fneg(\fdup(\BF_n))$, six bits above itself, and misclassifies it by Theorem~\ref{thm:classification-escape}. Every extension from inside comes with a new diagonal. Writing the truth table down does not remove the paradox; it moves it up.

\begin{theorem}[Quadratic lower bound]
\label{thm:lower-bound}
Every evaluator for level $n$ has level at least $(n-4)(n+2)$.
\end{theorem}

\begin{proof}
A term depends only on the positions it reads, and reading position $j$ costs $j+1$ bits. An evaluator for level $n$ must tell apart the quotations of $\mathsf{bit}_i$ and $\mathsf{bit}_{i+1}$, so it reads position $2(i+3)$ for every $i$ with $i+5\leq n$. Summing $2(i+3)+1$ over these $i$ gives $(n-4)(n+2)$.
\end{proof}

The quadratic term is the price of unary addressing. The optimal internal evaluator is bounded below by a quadratic; the one constructed here, $\BF_n$, has at least exponential level.

\subsection{The Concrete Hierarchy}
\label{sec:brute:hierarchy}

\begin{theorem}[Concrete hierarchy]
\label{thm:concrete-hierarchy}
For every oracle, the family $(V_n)_n$ satisfies Theorem~\ref{thm:hierarchy} unconditionally, with
\begin{equation}
	n_{k+1}=\lvl(V_{n_k})+9,
	\quad
	3^{n_k+1}+6\leq n_{k+1}.
\end{equation}
\end{theorem}

The equality is the nine-bit cost of forming the diagonal; the inequality is the size of brute force. Iterated, it gives tower growth.

\section{Coding and Oracles}
\label{sec:coding}

\subsection{Coding Sensitivity}
\label{sec:coding:sensitivity}

Keep terms and semantics fixed and replace the standard quotation by its bitwise complement.

\begin{theorem}[Coding-relative branch]
\label{thm:coding}
Under every quotation map, the diagonal claim is equivalent to its non-acceptance. The standard and the complemented quotation are both injective. The one-node verifier $\mathsf{bit}_0$ leaves a true claim uncertified under the first and accepts a false claim under the second, and no faithful interpretation that is the identity on terms relates the two.
\end{theorem}

The verifier reads the first bit of the quotation of its diagonal sentence. That sentence begins with a negation, whose tag starts with $0$; complementing the code turns it into $1$. The fixed point is robust; which branch occurs depends on the coding.

\subsection{Relativization}
\label{sec:coding:relativization}

\begin{theorem}[Relativization]
\label{thm:relativization}
Quotation, diagonal formation, $\BF_n$, the certifiers $V_n$ and the diagonal sentences are the same terms for every oracle, and the fixed point and the bounded family hold uniformly. Oracle-free terms ignore the oracle, and an oracle node distinguishes two oracles. For every Boolean oracle, every oracle circuit specializes to an oracle formula no larger than itself whose relative satisfiability is its acceptance, and relative truth tables decide it.
\end{theorem}

By \citet{bakergillsolovay75}, no argument of this kind settles $\mathcal{P}$ versus $\mathcal{NP}$ on its own. We return to this in Section~\ref{sec:discussion:about}.

\section{The Internal Reading}
\label{sec:internal}

Read acceptance as a provability operator $\Box$ in a Hilbert calculus. The same formation rules then produce a sentence that the calculus proves equivalent to its own non-provability.

\begin{proposition}[Internal fixed point]
\label{prop:internal}
A Hilbert calculus with modus ponens and the axioms $\textsc{K}\colon a\to(b\to a)$ and $\textsc{S}\colon(a\to(b\to c))\to((a\to b)\to(a\to c))$, a representation of $\Box$ in the language, and certified rules for negation and duplication yield the internal fixed point; the derivability conditions D1--D3 are not used. With D1--D3, stipulated as premises, and consistency, the calculus proves neither $G_v$ nor $\neg\Box\bot$. A consistent \textsc{GL} calculus and a consistent language model satisfy all these premises together.
\end{proposition}

The theorems of \citet{loeb} and the second incompleteness theorem \citep{goedel31} follow in the standard way \citep{boolos}. What is new is only the division of labour: forming the fixed point needs no derivability condition, and its incompleteness consequences need all three. None of this enters the SAT theorem.

\section{Mechanization}
\label{sec:mech}

\begin{figure}[H]
	\centering
	\vspace{1em}
	\includegraphics[width=3cm]{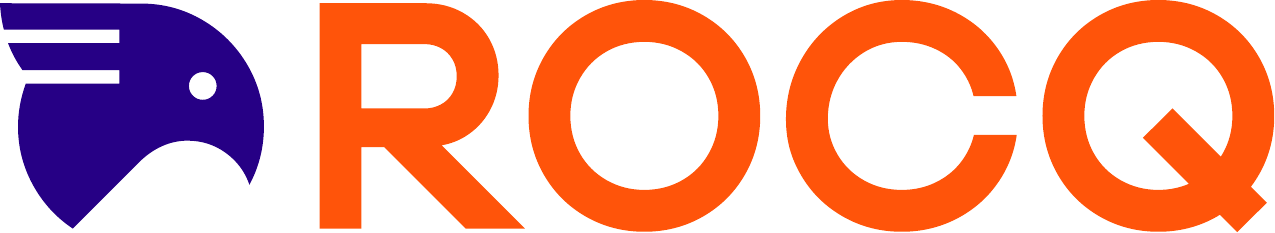}
	\vspace{1em}
	\caption{The 38 \texttt{L003} contracts are formalized in \textsc{Rocq}.}
	\label{fig:rocq}
\end{figure}

\begin{mechanization}
\label{mech:mech}
The development \texttt{L003} consists of 38 contracts in \textsc{Rocq}~9.1.1, frozen in \texttt{L003\_95\_API.v}. The aggregate \texttt{l003\_contract\_qed} proves all 38, and \texttt{l003\_core\_qed} proves the coding-free identity, the dichotomy, the size-free fixed point, brute-force evaluation and the plain-SAT theorem. Both assumption reports are closed. Proofs are intuitionistic; the computational layer, including quotation, $\BF_n$, the certifiers, truth tables, specialization and Tseitin conversion, is extracted to \textsc{OCaml}.
\begin{center}
	\href{https://github.com/Milan-Rosko/Proofcase/tree/main/theories/L003}{\texttt{GitHub $\to$ Milan-Rosko/Proofcase $\to$ L003}}
\end{center}
Every theorem in the text is stated with exactly the premises its endpoint uses. Semantics is proposition-valued throughout, and decidability enters only where a theorem names it. The table maps each result to its endpoint.
\begin{table}[H]
\centering
\renewcommand{\arraystretch}{1.25}
\scriptsize
\begin{adjustbox}{max width=\textwidth}
\begin{tabular}{@{}lll@{}}
	\toprule
	Result & Reference & Endpoint \\
	\midrule
	Identity (concrete) & \ref{lem:concrete-identity} & \texttt{certified\_coding\_free\_identity\_contract} \\
	Fixed point (concrete) & \ref{cor:concrete-fixed} & \texttt{relativization} \\
	Specialization & \ref{lem:specialization} & \texttt{specialization}; \texttt{classify\_correct} \\
	\textsc{Solver's Paradox} & \ref{thm:solvers-paradox} & \texttt{solvers\_paradox\_plain\_sat} \\
	Minimal branches & \ref{ex:minimal} & \texttt{certified\_minimal\_branches\_contract} \\
	CNF form & \ref{cor:cnf} & \texttt{certified\_tseitin\_contract} \\
	Identity and fixed point & \ref{thm:fixed-point} & \texttt{certified\_coding\_free\_identity\_contract}; \\
	 & & \texttt{certified\_self\_application\_contract} \\
	Pointwise incompleteness & \ref{cor:pointwise} & \texttt{certified\_pointwise\_incompleteness\_contract} \\
	Faithful invariance & \ref{prop:faithful} & \texttt{certified\_faithful\_invariance\_contract} \\
	Certification Dichotomy & \ref{thm:dichotomy} & \texttt{certified\_certification\_dichotomy\_contract} \\
	Sharpness & \ref{thm:sharpness} & \texttt{certified\_dichotomy\_sharpness\_contract} \\
	Constructive boundary & \ref{rem:boundary} & \texttt{constructive\_boundary} \\
	Classification escape & \ref{thm:classification-escape} & \texttt{diagonal\_escape}; \texttt{certified\_no\_internal\_classifier\_contract} \\
	Evaluator ascent & \ref{lem:evaluator-ascent} & \texttt{evaluator\_ascent}; \texttt{certification\_ascent} \\
	Settled one level up & \ref{thm:settled} & \texttt{certified\_settled\_one\_level\_up\_contract} \\
	Ascent hierarchy & \ref{thm:hierarchy} & \texttt{ascent\_hierarchy}; \texttt{ascent\_hierarchy\_bounds}; \texttt{proliferation} \\
	Concrete grading & \ref{thm:grading} & \texttt{certified\_concrete\_grading\_contract} \\
	Brute-force evaluation & \ref{thm:bf} & \texttt{certified\_brute\_force\_evaluation\_contract} \\
	The SAT claim, settled & \ref{cor:sat-settled} & \texttt{solvers\_paradox\_plain\_sat} \\
	Quadratic lower bound & \ref{thm:lower-bound} & \texttt{certified\_evaluator\_lower\_bound\_contract} \\
	Concrete hierarchy & \ref{thm:concrete-hierarchy} & \texttt{certified\_concrete\_hierarchy\_contract} \\
	Coding-relative branch & \ref{thm:coding} & \texttt{certified\_coding\_sensitivity\_contract} \\
	Relativization & \ref{thm:relativization} & \texttt{relativization}; \texttt{certified\_relativized\_sat\_contract} \\
	Internal fixed point & \ref{prop:internal} & \texttt{language\_fixed\_point}; \texttt{consistent\_certifies\_neither}; \\
	 & & \texttt{internal\_premises\_satisfiable}; \texttt{language\_premises\_satisfiable} \\
	\bottomrule
\end{tabular}
\end{adjustbox}
\caption{Results and their \textsc{Rocq} endpoints. The concrete fixed point is the oracle-uniform clause of \texttt{relativization}. The remaining contracts (\texttt{lawvere\_fixed\_point}, \texttt{no\_point\_surjection\_to\_bool}, the predicativity, size-free, truth-obstruction, one-semantics, \texttt{loeb} and \texttt{goedel\_second} contracts) underlie the text without a statement of their own.}
\label{tab:rocq-endpoints}
\end{table}
\end{mechanization}

\section{Discussion}
\label{sec:discussion}

\subsection{On the Result}
\label{sec:discussion:about}

Every result above holds for every oracle, and its proof is the same for every oracle: the terms $g_v$, $\BF_n$, $V_n$ and $\psi_v$ do not depend on $O$, and no proof decides anything about it (Theorem~\ref{thm:relativization}). The oracle is not invisible to the language, since an oracle node distinguishes two oracles; the constructions simply never take that information. This is the barrier of \citet{bakergillsolovay75} in a form that can be pointed at. A separation of $\mathcal{P}$ from $\mathcal{NP}$ holds relative to some oracles and fails relative to others, so no construction that carries the oracle without deciding anything about it can establish it, or its negation.

Read \emph{antirealistically}, in the sense of \citet{dummett75}, the same fact supports a stronger case: oracle-uniformity would constrain not only the methods that could settle $\mathcal{P}$ versus $\mathcal{NP}$ but the answer itself, which would have to consist of information specific to one representation. We leave that case to the reader.

The subject of the paper is certification: what a verifier represented in a language can accept about claims formed from it. On that subject the results are sharp. The identity is coding-free, the branch is coding-relative, the decision premise is exact, and the place of settlement is measured in bits. A resource-bounded version, in which quotation, duplication and projection preserve polynomial bounds, would connect the framework to $\mathcal{NP}$ verifiers. By Theorem~\ref{thm:relativization} it would still not separate classes by itself.

\subsection{Local Settlement and Internal Totality}
\label{sec:discussion:local}

Every represented claim is settled by some certifier of the language: the certifier for its level. No certifier of the language settles every claim, and by Theorem~\ref{thm:classification-escape} no solver of the language classifies every instance. The diagonal denies only the passage from the first statement to its uniform form. Settling each claim, and extending the rule to every level, require no new principle; what fails is the demand that the result be a single expression of the language.

\citet{post44} called a set creative when it is enumerable and its complement is productive: from any enumerable part of the complement, given by an index, one computes a point of the complement that the part misses. The map $v\mapsto g_v$ plays the same role for certification, with expressions in place of indices and nine bits in place of a computation. In complexity theory such a map is a \emph{refuter} \citep{chenjinsanthanamwilliams24}. The analogy is not exact, since our verifiers are finite and our hierarchy is unconditional. Any internal survey of the space comes with an effective way past it.

This is the vicious-circle principle of \citet{russell08}, that whatever involves all of a collection must not be one of the collection, in graded and mechanized form. Theorem~\ref{thm:classification-escape} and Lemma~\ref{lem:evaluator-ascent} turn the principle into a theorem with a price: an expression that evaluates a level lies above it, and no expression evaluates every level. If the classification of complexity classes is a practice of this kind, conducted in the language it classifies, the Zoo might grow for the same reason. The theorems prove that an expression evaluating a level lies above it and that none evaluates every level; that the Zoo grows for this reason is a reading of them.

In practice, the evaluator one level up is whoever is doing the mathematics. The theorems say only what happens once that position is itself written down.

\subsection{Open Problems}
\label{sec:discussion:open}

Every internal evaluator for level $n$ has level at least $(n-4)(n+2)$, and the evaluator constructed here, $\BF_n$, has level at least $3^{n+1}-3$. Whether some evaluator has polynomial level is open. Universal circuits \citep{valiant76} suggest a polynomial internal evaluator, but such an evaluator must simulate projection across alternations of quantifiers over witness streams; constructing one, or showing that none exists, would close the gap. A proof-checking calculus inside the concrete language, whose box is the projection of an actual verifier and which satisfies D1--D3, remains to be built. Finally, it remains to characterize the typed and resource-bounded languages that implement negation, diagonal duplication and witness projection faithfully; the results of \citet{brownpalsberg16} and \citet{dingzhang24} mark the two sides of that boundary.

\subsection{Summary}
\label{sec:discussion:summary}

Brute force is ascent. A truth table outside the language settles every propositional instance. Written inside it, the table is one more expression: it settles the level below from above and misclassifies its own diagonal instance. For every circuit verifier the certification claim returns to SAT as a formula no larger than the verifier, which the verifier misjudges and only a position above it classifies.

A decidable setting thus carries a Gödel-like certification obstruction without an undecidable proposition. Every diagonal claim of a circuit verifier is decided by a truth table; what fails is a single expression that evaluates every expression of its own language. The failure is graded in bits, and the line between proving disagreement and deciding its direction is exact: the first is constructive, the second, demanded of every verifier, is excluded middle.

If an obstruction of this kind reached $\mathcal{P}$ versus $\mathcal{NP}$, it could appear from inside only as a missing technique, since a sound verifier cannot certify its own diagonal failure (Corollary~\ref{cor:pointwise}): the claim that it accepts no witness for $g_v$ is $G_v$ itself. The observation cuts both ways admittedly: It answers the objection that the difficulty is only a lack of technique, since that is how a diagonal obstruction would look from inside; and it is no evidence for one, since an ordinary lack of technique looks the same. What could tell the two apart is a view from above, since a diagonal claim is settled one level up. Whether $\mathcal{P}$ versus $\mathcal{NP}$ is the diagonal claim of some position, and not only beyond the reach of the formation laws, is not shown here, and the ascent offers no final position whose claim it could be. We leave this question, too, to the reader.

\citet{rice53} shows that no machine decides a nontrivial semantic property of every machine. Humans have settled such properties for particular machines: the halting of every five-state machine has been settled with a mechanized proof \citep{bb5}, while six-state machines include ones whose halting is an open problem of Collatz type. Each case is local settlement in the sense of Section~\ref{sec:discussion:local}, and since every finite record of verdicts is computable, no record could settle the question
\begin{quotation}
	Was there ever a case in which a human decided a nontrivial semantic property, and not only instances of it?
\end{quotation}

\clearpage

\section*{References and Notes}

{\scriptsize
	\bibliographystyle{plainnat}
	\setlength{\bibsep}{0.1em}
	\bibliography{refs}}

\clearpage

\vspace*{\fill}

\subsection*{Note on the Second Version}

The present version makes representation explicit and proves the thesis. Its claim that every total SAT solver misclassifies an instance is now the classification half, restricted to solvers represented in the language (Theorem~\ref{thm:classification-escape}). The certification half, in which the misjudged claim returns to SAT as $\neg\SAT(\psi_v)$, is new. Diagnosis by a stronger theory has become settlement one level up, where the claim is decided and a fresh one recurs (Theorem~\ref{thm:hierarchy}). Naming the branch uniformly is now shown to be excluded middle (Theorem~\ref{thm:sharpness}).

The structural obstruction to $\mathcal{P}$ versus $\mathcal{NP}$ proposed there is withdrawn: it rested on representability and diagonalization alone, which hold for every oracle, so its claimed independence from relativization is reversed (Theorem~\ref{thm:relativization}).

\subsection*{Final Remarks}

The author welcomes criticism, proposed extensions, scholarly correspondence, and constructive dialogue. No conflicts of interest are declared. This research received no funding.

\begin{center}
	\scriptsize{
	\vspace{1em}
	Milan Rosko\\
	\vspace{1em}
	ORCID: \href{https://orcid.org/0009-0003-1363-7158}{\textsf{0009-0003-1363-7158}}\\[1ex]
	Email: \href{mailto:hi*atsymbol*milanrosko.com}{\textsf{hi*atsymbol*milanrosko.com}}\\[1ex]
	\vspace{1em}
	Licensed under \ccby \\ \href{http://creativecommons.org/licenses/by/4.0/}{\scriptsize\textsf{creativecommons.org/licenses/by/4.0}}
	}
\end{center}
\vspace*{\fill}